\documentclass[conference]{IEEEtran}%
\usepackage{amsfonts}
\usepackage{amsmath}
\usepackage{amssymb}
\usepackage{graphicx}%
\providecommand{\U}[1]{\protect\rule{.1in}{.1in}}
\newtheorem{theorem}{Theorem}

\newtheorem{definition}{Definition}

\newtheorem{lemma}{Lemma}

\newtheorem{proposition}{Proposition}
\newtheorem{remark}{Remark}

\begin{document}

\title{Distributed Estimation in Multi-Agent Networks}

%

%TCIMACRO{\TeXButton{Author Information}{\author{\authorblockN
%{Lalitha Sankar and H. Vincent Poor}
%\authorblockA{Dept. of Electrical Engineering,\\
%Princeton University,
%Princeton, NJ 08544.\\
%{\{lalitha,poor\}}@princeton.edu\\}}}}%
%BeginExpansion
\author{\authorblockN{Lalitha Sankar and H. Vincent Poor}
\authorblockA{Dept. of Electrical Engineering,\\
Princeton University,
Princeton, NJ 08544.\\
{\{lalitha,poor\}}@princeton.edu\\}}%
%EndExpansion
%

%TCIMACRO{\TeXButton{Make Title}{\maketitle}}%
%BeginExpansion
\maketitle
%EndExpansion
%

%TCIMACRO{\TeXButton{Begin abstract}{\begin{abstract}}}%
%BeginExpansion
\begin{abstract}%
%EndExpansion
\footnotetext{The research was supported in part by the AFOSR under MURI Grant
FA9550-09-1-0643 and in part by the NSF under Grant CCF-10-16671.
\par
{}}A problem of distributed state estimation at multiple agents that are
physically connected and have competitive interests is mapped to a distributed
source coding problem with additional privacy constraints. The agents interact
to estimate their own states to a desired fidelity from their (sensor)
measurements which are functions of both the local state and the states at the
other agents. For a Gaussian state and measurement model, it is shown that the
sum-rate achieved by a distributed protocol in which the agents broadcast to
one another is a lower bound on that of a centralized protocol in which the
agents broadcast as if to a virtual CEO converging only in the limit of a
large number of agents. The sufficiency of encoding using local measurements
is also proved for both protocols.%

%TCIMACRO{\TeXButton{End abstract}{\end{abstract}}}%
%BeginExpansion
\end{abstract}%
%EndExpansion

\section{Introduction}

We consider a network of $K$ distributed agents in which each agent observes
sensor measurements from a distinct part of a large interconnected physical
network. Examples of such networks include cyber-physical systems,
specifically the smart grid, in which an agent can be viewed as a regional
operator whose power measurements are affected by those at other agents due to
the physical grid connectivity. Agent $k$ is interested in estimating the
state (defined as a set of system parameters; for e.g., voltages and phases in
the electric grid) of its local network from its measurements, $Y_{k},$ which
are a function of both the local state $X_{k}$ and the states $X_{l},$
$l\not =k,$ $l,k\in\left\{  1,2,\ldots,K\right\}  $ of other agents in the
network where the states $X_{k}$ are assumed to be independent of each other.

Estimating $X_{k}$ at agent $k$ with high fidelity requires the agents to
interact and share data amongst themselves. While the estimate fidelity is
crucial to the control decisions made by the agents, in many distributed
systems, for competitive reasons, the agents wish to keep their state
information private. This leads to a problem of \textit{competitive privacy}
which captures the tradeoff between the utility to the agent (estimate
fidelity) that can be achieved via cooperation and the resulting privacy
leakage (quantified via mutual information).

Mapping utility to distortion and privacy to leakage quantified via mutual
information, one can abstract the competitive privacy problem as a distributed
source coding problem with additional leakage constraints. The set of all
achievable rate-fidelity-leakage tuples determines the utility-privacy
tradeoff region. In \cite{LS_SG1}, we introduced and studied this problem for
a two-agent interactive system with Gaussian states and noisy Gaussian
measurements. We proved that side-information (measurements at the other
agent) aware Wyner-Ziv encoding \cite{Wyner_Ziv} at each agent achieves both
the minimal rate and the minimal leakage for every choice of fidelity
(quantified via mean-squared distortion).

Even without additional privacy constraints, the problem of determining the
set of all rate-distortion tuples in a multi agent network is related to the
distributed source coding problem \cite{DSS5,DSS5a} which remains open.
Furthermore, for a relatively simpler setting obtained by assuming that a
central entity, often referred to as a chief executive officer (CEO), wishes
to estimate the states $X_{k},$ for all $k,$ from the transmissions of all
agents, we obtain a multi-variate (vector) Gaussian CEO problem which also
remains open except for specific cases \cite{DSS5b}.

Circumventing these challenges, we focus on the rate-distortion-leakage
behavior in the limit of large $K$ for a \textit{distributed protocol }in
which each agent encodes its measurements taking into account the prior
broadcasts of the other agents (henceforth referred to as \textit{progressive
encoding}) as well as the side-information at the other agents. We compare the
performance of this protocol with a \textit{centralized} \textit{protocol} in
which the agents broadcast their encoded messages as if to a virtual CEO. We
consider a noisy Gaussian measurement model at each agent with the same level
of interference from the states of the other agents. For this symmetric
\ model, our results demonstrate that the sum-rate achieved by distributed
protocol outperforms that for the centralized schemes with asymptotic
convergence with $K$. We also prove the sufficiency of encoding local
measurements for both protocols and present outer bounds for the per user rate
and leakage.

The paper is organized as follows. We introduce the model and communication
protocols in Section \ref{Sec_II}. In Section \ref{Sec_III} we develop the
achievable rate-distortion-leakage tuples for both protocols as well as outer
bounds. We conclude in Section \ref{Sec_CR}.

\section{\label{Sec_II}Preliminaries}

\subsection{Model and Metrics}

We consider a network of $K$ agents such that, at any time instant $i,$
$i=1,2,\ldots,n,$ the measurement $Y_{k,i}$ at agent $k$, $k=1,2,\ldots,K,$ is
related to the states $X_{m,i},$ $m=1,2,\ldots,K,$ at the agents as follows:%
\begin{equation}
Y_{k,i}=X_{k,i}+\sum\limits_{l=1,l\not =k}^{K}\sqrt{h}X_{l,i}+Z_{k,i},\text{
}k=1,2,\ldots,K, \label{Model}%
\end{equation}
where the state variables $X_{m,i}\sim\mathcal{N}(0,\sigma^{2})$, for all $m$
and $i$ are assumed to be independent and identically distributed (i.i.d.) and
are also independent of the i.i.d. noise variables $Z_{k,i}\sim\mathcal{N}%
(0,1)$. The coefficient $h>0$ is assumed to be fixed for all time and known at
all agents. We assume that the $k^{th}$ agent observes a sequence of $n$
measurements $Y_{k}^{n}=[Y_{k,1}$\textbf{ }$Y_{k,2}$ $\ldots$ $Y_{k,n}]$, for
all $k$, prior to communications.

\textit{Utility}: For the continuous Gaussian distributed state and
measurements, a reasonable metric for utility at the $k^{th}$ agent is the
mean square error $D_{k}$ between the original and the estimated state
sequences $X_{k}^{n}$ and $\hat{X}_{k}^{n}$, respectively.

\textit{Privacy}: The measurements at each agent in conjunction with the
quantized data shared by the other agents while enabling accurate estimation
also leaks information about the other agents' states. We capture this leakage
using mutual information.

\subsection{Communication Protocol}

We assume that each agent broadcasts a function of its measurements
(\textit{distributed procotol}) to all agents and they do so in a round-robin
fashion. We assume that all agents encode in one of the following two ways: i)
\textit{local encoding} in which each agent quantizes only its measurements;
or ii) \textit{progressive encoding} in which each agent encodes and transmits
taking into account both its measurements and prior communications from other
agents. In both cases, the agents transmit at a rate that takes into account
the correlated measurements and prior communications of other agents.

To better understand the advantage of the above distributed procotol, we also
consider the case where the agents broadcast as if communicating with a
virtual central operator, say CEO, henceforth referred to as the
\textit{centralized protocol}. This may be viewed as the case in which the
computing power at the agents is limited and the CEO shares with each agent
its received messages (which are then decoded at each agent). For either
protocol, the encoding can be either local or progressive. Let $I_{p}%
\in\left\{  0,1\right\}  $ and $I_{enc}\in\left\{  0,1\right\}  $ be random
variables that denote the choice of protocols and encodings such that
$I_{p}=1$ and $I_{p}=0$ for the distributed and centralized protocol,
respectively, and $I_{enc}=1$ and $I_{enc}=0$ for the progressive and local
encoding, respectively.

Formally, the encoder at agent $k$ maps its measurements to an index set
$\mathcal{J}_{k}$ where
\begin{equation}
\mathcal{J}_{k}\equiv\left\{  1,2,\ldots,J_{k}\right\}  \text{, }%
k=1,2,\ldots,K, \label{Enc1}%
\end{equation}
is the index set at the $k^{th}$ agent for mapping the measurement sequence,
and the prior communications (progressive encoding), via the encoder $f_{k}$,
$k=1,2,\ldots K,$ defined as%
\begin{equation}
f_{k}:\mathcal{Y}_{k}^{n}\times I_{enc}\cdot%
%TCIMACRO{\tprod \nolimits_{l=1}^{k-1}}%
%BeginExpansion
{\textstyle\prod\nolimits_{l=1}^{k-1}}
%EndExpansion
\mathcal{J}_{l}\rightarrow\mathcal{J}_{k}, \label{Enc2}%
\end{equation}
such that at the end of the $K$ broadcasts, one from each agent, the decoding
function $F_{k}$ at the $k^{th}$ agent (or the CEO) is a mapping from the
received message sets (both protocols) and the measurements (the distributed
procotol) to that of the reconstructed sequence denoted as%
\begin{equation}
F_{k}:\mathcal{J}_{1}\times\ldots\times\mathcal{J}_{K}\times\left(
\mathcal{Y}_{k}^{n}\cdot I_{p}\right)  \rightarrow\mathcal{\hat{X}}_{k}%
^{n},\text{ \ }k=1,2,\ldots,K. \label{Dec}%
\end{equation}
Let $M_{k}$ denotes the size of $J_{k}$. The expected distortion $D_{k}\ $at
the $k^{th}$ agent is given by
\begin{equation}
D_{k}=\frac{1}{n}\mathbb{E}\left[
%TCIMACRO{\tsum \limits_{i=1}^{n}}%
%BeginExpansion
{\textstyle\sum\limits_{i=1}^{n}}
%EndExpansion
\left(  X_{k,i}-\hat{X}_{k,i}\right)  ^{2}\right]  \text{, }k=1,2,\ldots K,
\label{Dist}%
\end{equation}
The privacy leakage$,$ $L_{k}^{\left(  l\right)  }$, about state $k$ at agent
$l,$ $l\not =k,$ is given by%

\begin{equation}
L_{k}^{\left(  l\right)  }=\frac{1}{n}I\left(  X_{k}^{n};J_{1},J_{2}%
,\ldots,J_{K},Y_{l}^{n}\right)  ,\text{ for all }k\not =l.
\end{equation}
The communication rate of the $k^{th}$ agent is denoted by
\begin{equation}
R_{k}=n^{-1}\log_{2}M_{k},k=1,2,\ldots,K. \label{CommRate}%
\end{equation}

\begin{definition}
\label{DefUP}The utility-privacy tradeoff region is the set of all
$(D_{1},\ldots,D_{k},L_{1}^{\left(  2\right)  },\ldots,L_{1}^{\left(
K\right)  },\ldots,L_{K}^{\left(  1\right)  },\ldots,L_{K}^{(K-1)})$ for which
there exists a coding scheme given by (\ref{Enc1})-(\ref{Dec}) with parameters
$(n,K,M_{1},M_{2},D_{1}+\epsilon,\ldots,D_{K}+\epsilon,L_{1}+\epsilon
,\ldots,L_{K}+\epsilon)$ for $n$ sufficiently large such that $\epsilon
\rightarrow0$ as $n\rightarrow\infty$.
\end{definition}

\section{\label{Sec_III}Main Results}

We use the following proposition, lemma, and function definition in the sequel
to compute the achievable distortions and rates.

\begin{proposition}
\label{Prop1}For (column) vectors $\underline{A}$ and $\underline{B}$, let
$K_{\underline{A}\underline{A}}=var\left(  \underline{A}\right)  =E\left[
\left(  \underline{A}-E\left[  \underline{A}\right]  \right)  \left(
\underline{A}^{T}-E\left[  \underline{A}^{T}\right]  \right)  \right]  $ and
$K_{\underline{A}\underline{B}}=E\left[  \left(  \underline{A}-E\left[
\underline{A}\right]  \right)  \left(  \underline{B}^{T}-E\left[
\underline{B}^{T}\right]  \right)  \right]  $ denote the covariance and
cross-correlation matrices, respectively. The conditional variance
$E[var(\underline{A}|\underline{B})]$ is then given as $E[var(\underline
{A}|\underline{B})]=K_{\underline{A}\underline{A}}-K_{\underline{A}%
\underline{B}}K_{\underline{B}\underline{B}}^{-1}K_{\underline{A}\underline
{B}}^{T}.$
\end{proposition}

\begin{lemma}
\label{Lemma1}For a $K\times K$ symmetric Toeplitz matrix whose diagonal
entries are all $a,$ and off-diagonal entries are all $b$ the determinant is
$\left(  a+\left(  K-1\right)  b\right)  \left(  a-b\right)  ^{\left(
K-1\right)  }.$
\end{lemma}

\begin{proof}
The determinant is obtained by the following two operations: i) add columns
2-$K$ to column 1, and ii) subtract row $1$ from each of the remaining rows.
\end{proof}

\begin{definition}
For some $\alpha$, $\beta\in\mathcal{R}^{+},$ the function $f_{1}\left(
k,c\right)  \equiv\alpha+\left(  k-2\right)  \beta-\left(  k-1\right)  c$
varies over $k\in\left[  1,K\right]  $ and $c\in\mathcal{R}^{+}.$
\end{definition}

\subsection{Distortion}

We assume that each agent has the same distortion constraint $D$. The
distortion $D$ at each agent ranges from a minimum achieved when it has
perfect access to the measurements at all agents to a maximum achieved when it
estimates using only its own measurements. From the symmetry of the model in
(\ref{Model}), the minimal (resp. maximal) distortion achieved at each agent
is the same. Let $D_{\min}$ and $D_{\max}$ denote the minimal and maximal
distortions, respectively, at each agent. For the Gaussian model considered
here with minimum mean square error (MSE) constraints, we have
\begin{align}
D_{\min}  &  =E\left[  var(X_{1}|Y_{1}Y_{2}\ldots Y_{K})\right]  ,\text{
and}\label{Dmin_def}\\
D_{\max}  &  =E\left[  var(X_{1}|Y_{1})\right]  . \label{Dmax_def}%
\end{align}
We now determine $D_{\min}$ and $D_{\max}$. Let
\begin{subequations}
\label{alp_beta_defs}%
\begin{align}
\alpha &  \equiv E(Y_{l}^{2})=\sigma_{X}^{2}\left(  1+h\left(  K-1\right)
\right)  +1,\text{ for all }l\\
\beta &  \equiv E(Y_{l}Y_{k})=\sigma_{X}^{2}\left(  2\sqrt{h}+h\left(
K-2\right)  \right)  ,\text{ }l\not =k.
\end{align}
Note that for large $K$, $\alpha\rightarrow h\left(  K-1\right)  \sigma
_{X}^{2},$ and $\beta\rightarrow h\left(  K-2\right)  \sigma_{X}^{2}.$

\textit{Computation of }$D_{\max}$: Expanding (\ref{Dmax_def}), we obtain
\end{subequations}
\begin{equation}
D_{\max}=E\left[  var(X_{1}|Y_{1})\right]  =\sigma_{X}^{2}\left(
1-\frac{\sigma_{X}^{2}}{\alpha}\right)  . \label{Dmax}%
\end{equation}
For large $K,$ $D_{\max}\rightarrow\sigma_{X}^{2}$.

\textit{Computation of }$D_{\min}$: Expanding (\ref{Dmin_def}), we have
\begin{align}
D_{\min}  &  =E\left[  var(X_{1}|Y_{1}Y_{2}\ldots Y_{K})\right] \\
&  =\frac{\left\vert E\left[  var(X_{1}Y_{2}\ldots Y_{K}|Y_{1})\right]
\right\vert }{\left\vert E\left[  var(Y_{2}\ldots Y_{K}|Y_{1})\right]
\right\vert } \label{D_max2}%
\end{align}
where the simplification in (\ref{D_max2}) results from the assumption of
jointly Gaussian random variables. Applying Lemma \ref{Lemma1}, for
\begin{align}
c_{1}  &  =\sigma_{X}^{2}-\sigma_{X}^{4}/\alpha,\text{ }c_{2}=\sigma_{X}%
^{2}\left(  \sqrt{h}-\beta/\alpha\right)  ,\\
c_{3}  &  =\alpha-\beta^{2}/\alpha,\text{ and }c_{4}=\beta-\beta^{2}/\alpha,
\end{align}
we obtain the minimum distortion $D_{\min}$ as%
\begin{equation}
D_{\min}=D_{\max}\left(  1-\frac{\left(  K-1\right)  \frac{\sigma_{X}%
^{2}\left(  \sqrt{h}-\beta/\alpha\right)  ^{2}}{\left(  1-\sigma_{X}%
^{2}/\alpha\right)  }}{f_{1}\left(  K,\beta^{2}/\alpha\right)  }\right)  .
\label{Dmin_final}%
\end{equation}

\begin{remark}
For $K\rightarrow\infty,$ $D_{\min}\rightarrow D_{\max}(1-(1-\sqrt{h})^{2}/h)$.
\end{remark}

\subsection{Distributed Protocol}

A general coding strategy for this distributed source coding problem needs to
take into account: a) the order of agent broadcasts; b) multiple encoding
possibilities at each agent depending on whether the received data is used
alongwith local measurements in encoding; c) exploiting the correlated
measurements at other agents in broadcasting just sufficient data for other
agents to achieve their distortions; and d) multiple rounds of interactions.
We present a distributed encoding scheme with a single round of communication
(for simplicity of analysis) in which the agents broadcast in order (the
source permutation choice is irrelevant due to the symmetry of the model). The
local and progressive coding schemes differ in including the received data in
encoding at each agent, while the centralized and distributed protocols differ
in whether they exploit the correlated measurements at the other agents.

The achievable distortion $D$ in general depends on the encoding scheme
chosen. Let $R_{k}$ and $\tilde{R}_{k}$ denote the rates for the local and
progessive encoding schemes, respectively. We first consider the progressive
encoding scheme in which each agent broadcasts (to all other agents) a noisy
function of both its measurements and prior communications. More precisely,
agent $k$ maps its measurement and prior communication sequences to one among
a set of $2^{n\tilde{R}_{k}}$ $\tilde{U}_{k}^{n}$ sequences chosen to satisfy
the distortion constraints. The $\tilde{U}_{k}^{n}$ sequences are generated
via an i.i.d distribution of $\tilde{U}_{k,i}$ for all $i$ such that
$\tilde{U}_{1,i}=Y_{1,i}+Q_{1,i}$ and for all $k>1,$ $\tilde{U}_{k,i}=Y_{k,i}+%
%TCIMACRO{\tsum \nolimits_{l=1}^{k-1}}%
%BeginExpansion
{\textstyle\sum\nolimits_{l=1}^{k-1}}
%EndExpansion
a_{k,l}\tilde{U}_{l,i}+Q_{k,i}$ where $a_{k,l}\in\mathcal{R}$, and
$Q_{k,i}\sim N\left(  0,\sigma_{Q}^{2}\right)  $ is independent of $Y_{k,i}$
for all $k=1,2,\ldots,K,$ and $i=1,2,\ldots,n.$

The achievable distortion $D$ at agent $k$ as a result of estimating its state
using both its measurements $Y_{k}^{n}$ and the received sequences $\tilde
{U}_{l}^{n},$ for all $l\not =k,$ is such that $D\in\left[  D_{\min},D_{\max
}\right]  $ where $D_{\max}$ is achieved when $U_{l}^{n}=0$ for all $l$ and
$D=D_{\min}$ for $\sigma_{Q}^{2}=0$. On the other hand, for the local encoding
scheme, let $U_{k,i}=Y_{k,i}+Q_{k,i},$ for all $k$ and $i,$ such that agent
$k$ maps \textit{only} its measurement sequences to one among a set of
$2^{nR_{k}}$ $U_{k}^{n}$ sequences chosen to satisfy the distortion constraints.

\begin{theorem}
The sets $\mathcal{D}$ of all achievable distortions $D$ for the local and
progressive encoding schemes for the distributed protocol are the same.
\end{theorem}

\begin{proof}
For Gaussian codebooks and Gaussian measurements and from symmetry of the
model, the distortion $D$ at each agent is given by
\begin{align}
D  &  =\mathbb{E}\left[  var\left(  X_{1}|Y_{1}\tilde{U}_{1}\tilde{U}%
_{2}\tilde{U}_{3}\ldots\tilde{U}_{K}\right)  \right] \label{Dist_Prog}\\
&  =\mathbb{E}\left[  var\left(  X_{1}|Y_{1}U_{1}U_{2}U_{3}\ldots
U_{K}\right)  \right]  \in\lbrack D_{\min},D_{\max}] \label{Dist_LP1}%
\end{align}
where in (\ref{Dist_Prog}) we have used that fact that $\tilde{U}_{1}=U_{1},$
and conditioned on $U_{1},$it suffices to condition on $U_{2},$ and similarly
for the remaining $U_{k},$ $k>2$.
\end{proof}

\textit{Computation of }$D$: Using the independence of the quantization noise
$Q_{k}$ for all $k,$ as well as the independence of $Q_{k}$ and $X_{k}$, we
have $E\left[  U_{k}U_{l}\right]  =E\left[  Y_{k}Y_{l}\right]  =\beta$ for all
$l\not =k$ and $E\left[  U_{k}^{2}\right]  =E\left[  Y_{k}^{2}\right]
+E\left[  Q_{k}^{2}\right]  =\alpha+\sigma_{Q}^{2}.$ Thus, $D$ is obtained in
a manner analogous to the calculation of $D_{\min}$ with the replacement of
$c_{3}$ by $c_{3}+\sigma_{Q}^{2}$. Thus, we have%
\begin{equation}
D=D_{\max}\left(  1-\frac{\left(  K-1\right)  \frac{\sigma_{X}^{2}\left(
\sqrt{h}-\beta/\alpha\right)  ^{2}}{\left(  1-\sigma_{X}^{2}/\alpha\right)  }%
}{f_{1}\left(  K,\frac{\beta^{2}}{\alpha}\right)  +\sigma_{Q}^{2}}\right)  .
\end{equation}

\textit{Rate Computation}: We consider a round-robin protocol in which agent 1
broadcasts a quantized function of its measurements and prior communications
at a rate which takes into account all the side information at all other
agents. Thus, the rate $\tilde{R}_{1}$ required is the maximal of the rates
required to each agent and is given by
\begin{subequations}
\begin{align}
\tilde{R}_{1}  &  \geq I(\tilde{U}_{1};Y_{1})-\min\left(  I(\tilde{U}%
_{1};Y_{2}),\ldots,I(\tilde{U}_{1};Y_{K})\right) \\
&  =I(U_{1};Y_{1})-I(U_{1};Y_{2})=R_{1} \label{Rate_R1}%
\end{align}
where (\ref{Rate_R1}) follows from the symmetry of the measurement model, the
fact that $\tilde{U}_{1}=U_{1},$ and $R_{1}$ is the minimal rate required at
agent 1 for the local scheme. Next, agent 2 analogously broadcasts a function
of its measurements at a rate $R_{2}$ given by
\end{subequations}
\begin{subequations}
\begin{align}
\tilde{R}_{2}  &  \geq I(\tilde{U}_{2};Y_{2}\tilde{U}_{1})-\min_{l\in\left\{
1,...,K\right\}  ,l\not =2}I(\tilde{U}_{2};Y_{l}\tilde{U}_{1})\\
&  =I(\tilde{U}_{2};Y_{2}|\tilde{U}_{1})-\min_{l\in\left\{  1,...,K\right\}
,l\not =2}I(\tilde{U}_{2};Y_{1}|\tilde{U}_{1})\\
&  =I(U_{2};Y_{2})-I(U_{2};Y_{1})=R_{2} \label{Rate_R2}%
\end{align}
where (\ref{Rate_R2}) follows from $h(\tilde{U}_{2}|Y_{1}\tilde{U}%
_{1})-h(\tilde{U}_{2}|Y_{2}\tilde{U}_{1})=h(U_{2}|Y_{1})-h(U_{2}|Y_{2})$ since
$U_{2}-Y_{2}-U_{1}$ form a Markov chain and due to the symmetry of the model.
It can be verified easily that the bound in (\ref{Rate_R2}) is the minimal
rate $R_{2}$ for the local encoding scheme. One can similarly show that the
rate at which agent 3 broadcasts is
\end{subequations}
\begin{subequations}
\begin{align}
\tilde{R}_{3}  &  \geq I(\tilde{U}_{3};Y_{3}\tilde{U}_{1}\tilde{U}_{2}%
)-\min_{l\in\left\{  1,...,K\right\}  ,l\not =3}I(\tilde{U}_{3};Y_{1}\tilde
{U}_{1}\tilde{U}_{2})\\
&  =I(U_{3};Y_{3})-I(U_{3};Y_{1}U_{2})=R_{3}%
\end{align}
where we have used the fact that $U_{3}-Y_{3}-U_{1}U_{2}$ and $U_{1}%
-Y_{1}-U_{3}$ form Markov chains. Generalizing we have, for all $k>1,$%
\end{subequations}
\begin{subequations}
\begin{equation}
\tilde{R}_{k}=R_{k}\geq I(U_{k};Y_{k})-I(U_{k};Y_{1}U_{1}\ldots U_{k-1}%
),\text{ } \label{Rate_Rk}%
\end{equation}
where the bound in (\ref{Rate_Rk}) is the minimal rate at which agent $k$ is
required to broadcast when it only encodes $Y_{k}^{n}$.

\textit{Calculation of Leakage}: For the proposed progressive encoding, the
leakage of the state of agent $k$ at any other agent $j\not =k,$ for all such
$k,j,$ is bounded as
\end{subequations}
\begin{subequations}
\label{Leakage}%
\begin{align}
L_{k}^{(j)}  &  =\frac{1}{n}I(X_{k}^{n};Y_{j}^{n}J_{1}J_{2}\ldots
J_{K}),\text{ }j\not =k\\
&  \geq I(X_{1};Y_{2}\tilde{U}_{1}\ldots\tilde{U}_{K})=I(X_{1};Y_{2}%
U_{1}\ldots U_{K})\label{Leakage_2}\\
&  =\frac{1}{2}\log\left(  \frac{\alpha f_{1}\left(  K,\beta^{2}%
/\alpha\right)  }{\left(  \alpha-\sigma_{X}^{2}\right)  f_{1}\left(
K,c_{5}\right)  }\right)  \label{Leakage_3}%
\end{align}
where (\ref{Leakage_2}) is a result of the model symmetry, the code
construction and typicality arguments and is omitted for brevity. The bound in
(\ref{Leakage_3}) follows from the relation of the code constructions for the
two encoding schemes and $c_{5}=\left.  (\beta-\sqrt{h}\sigma_{x}^{2}%
)^{2}\right/  \left(  \alpha-\sigma_{x}^{2}\right)  +h\sigma_{X}^{2}$.
\end{subequations}
\begin{theorem}
\label{Lemma_3}It is sufficient to encode the local measurements at each agent
in the distributed protocol.
\end{theorem}

Theorem \ref{Lemma_3} follows directly from the fact that for Gaussian
encoding, from (\ref{Dist_LP1}), (\ref{Rate_Rk}), and (\ref{Leakage_3}), we
have that the set of all rate-distortion-leakage tuples achieved by the local
and progressive encoding schemes is the same.

The sum-rate of the distributed scheme $R_{sum}^{Dist}=\sum\nolimits_{k=1}%
^{K}R_{k}$ can be simplified as
\begin{subequations}
\label{Rsum}%
\begin{align}
R_{sum}^{Dist}  &  =h\left(  U_{2}U_{3}\ldots U_{K}|Y_{1}\right)
+h(U_{1}|Y_{2})-\frac{K}{2}\log\left(  2\pi e\sigma_{Q}^{2}\right)
\label{RsumDist5}\\
&  =\frac{K}{2}\log\left(  \frac{\alpha+\sigma_{Q}^{2}-\beta}{\sigma_{Q}^{2}%
}\right)  +\frac{1}{2}\log\left(  \frac{\left(  \alpha+\sigma_{Q}^{2}%
-\frac{\beta^{2}}{\alpha}\right)  }{\left(  \alpha+\sigma_{Q}^{2}%
-\beta\right)  }\right) \label{RsumDistFin}\\
&  \text{ \ \ }+\frac{1}{2}\log\left(  \left.  (f_{1}\left(  K,\beta
^{2}/\alpha\right)  +\sigma_{Q}^{2})\right/  \left(  \alpha+\sigma_{Q}%
^{2}-\beta\right)  \right) \nonumber
\end{align}
where (\ref{RsumDistFin}) is obtained from (\ref{RsumDist5}) by determining
$\left\vert E\left[  var\left(  \underline{U}_{K}|Y_{1}\right)  \right]
\right\vert $ where $\underline{U}_{K-1}=\left[  U_{2}\text{ }U_{3}\text{
}\ldots\text{ }U_{K}\right]  ^{T}$ denotes a column vector of length $\left(
K-1\right)  $. By expanding $E\left[  var\left(  \underline{U}_{K-1}%
|Y_{1}\right)  \right]  $ using Proposition \ref{Prop1}, one can verify that
$\left\vert E\left[  var\left(  \underline{U}_{K}|Y_{1}\right)  \right]
\right\vert $ simplifies to finding the determinant of the $\left(
K-1\right)  \times\left(  K-1\right)  $ Toeplitz matrix with diagonal and off
diagonal entries $\alpha+\sigma_{Q}^{2}-\frac{\beta^{2}}{\alpha}$ and
$\beta-\frac{\beta^{2}}{\alpha},$ respectively, which from Lemma \ref{Lemma1}
is given by $f_{1}\left(  K,\beta^{2}/\alpha\right)  (\alpha+\sigma_{Q}%
^{2}-\beta)^{\left(  K-2\right)  }.$ One can similarly show that $E\left[
var\left(  U_{1}|Y_{2}\right)  \right]  =\alpha+\sigma_{Q}^{2}-\beta
^{2}/\alpha.$

In the limit of $K\rightarrow\infty,$ $\left(  K-2\right)  \beta-\left(
K-1\right)  \frac{\beta^{2}}{\alpha}\rightarrow0,$ $\alpha-\beta^{2}%
/\alpha\rightarrow h,$ $\alpha-\beta\rightarrow h$, and therefore, the second
and third log terms in (\ref{RsumDistFin}) scale as $\log\left(  K\right)  .$
Thus, in the limit, the per agent rate $R=R_{sum}^{Dist}/K$ is given by
\end{subequations}
\begin{equation}
\lim_{K\rightarrow\infty}R=\frac{1}{2}\log\left(  \frac{\alpha+\sigma_{Q}%
^{2}-\beta}{\sigma_{Q}^{2}}\right)  .
\end{equation}

\subsection{Distributed vs. Centralized}

We now compare the distributed protocol to a centralized protocol in which
each agent broadcasts at a rate intended for a (virtual) CEO, and thus, is
oblivious of the correlated measurements at the other agents. Here again, the
agents can use a progressive encoding scheme analogously to the distributed
protocol. As in the distributed protocol, here too one can show that a local
encoding scheme suffices, in which agent $k$ generates a codebook $U_{k}^{n}$
whose entries $U_{k,i}$ are generated in an i.i.d fashion such that
$U_{k,i}=Y_{k,i}+Q_{k,i}$, $Q_{k,i}$ is independent of $Y_{k,i}$ and
$Q_{l,i},$ for all $l\not =k,$ for all $k,$ and for all $i.$ The compression
rates are bounded as follows. First, agent $1$ transmits its quantized
measurements at a rate $R_{1}$ such that for error-free decoding of $U_{1}%
^{n}$ at the decoder, we require%
\begin{equation}
R_{1}\geq I\left(  U_{1};Y_{1}\right)  .
\end{equation}
Agent 2 takes into account the knowledge of $U_{1}^{n}$ at all agents and
broadcasts at a rate%
\begin{equation}
R_{2}\geq I\left(  U_{2};Y_{2}\right)  -I\left(  U_{2};U_{1}\right)  .
\end{equation}
Note that the agents broadcast taking into account the prior transmissions (as
if to a CEO) but not the side information at the other agents. Continuing
similarly, we have for all $k\geq2$,%
\begin{equation}
R_{k}\geq I\left(  U_{2};Y_{k}\right)  -I\left(  U_{k};U_{1}U_{2}\ldots
U_{k-1}\right)  .
\end{equation}
The resulting sum rate $R_{sum}^{CEO}=\sum\nolimits_{k=1}^{K}R_{k}$ can be
simplified as
\begin{align}
R_{sum}^{CEO}  &  =\sum\nolimits_{k=1}^{K}I(U_{k};Y_{k})-\sum\nolimits_{k=2}%
^{K}I(U_{k};U_{1}\ldots U_{k-1})\\
&  =h\left(  U_{K},U_{K-1}\ldots U_{1}\right)  -\frac{K}{2}\log\left(  2\pi
e\sigma_{Q}^{2}\right) \\
&  =\frac{K}{2}\log\left(  \frac{\left(  \alpha+\sigma_{Q}^{2}-\beta\right)
}{\sigma_{Q}^{2}}\right) \label{RCEOsum}\\
&  \text{ \ \ }+\frac{1}{2}\log\left(  \frac{\left(  \alpha+\sigma_{Q}%
^{2}+\left(  K-1\right)  \beta\right)  }{\left(  \alpha+\sigma_{Q}^{2}%
-\beta\right)  }\right)  .\nonumber
\end{align}
Thus, the rate on average per user is $R^{CEO}=R_{sum}^{CEO}/K$ which
converges in the limit of a large number of agents $K$ to
\begin{equation}
\lim_{K\rightarrow\infty}R^{CEO}=\frac{1}{2}\log\left(  \frac{\left(
\alpha+\sigma_{Q}^{2}-\beta\right)  }{\sigma_{Q}^{2}}\right)  .
\end{equation}

Comparing (\ref{RsumDistFin}) and (\ref{RCEOsum}), we can verify that for
every choice of $\sigma_{Q}^{2},$ and hence $D,$ $R_{sum}^{CEO}>R_{sum}%
^{Dist}$. Furthermore, one can also show that the leakage at each agent for
the centralized protocol is the same as the distributed protocol in
(\ref{Leakage}) and is the same for both the local and progressive encoding
schemes. The following theorem summarizes our results.

\begin{theorem}
The average per user rate of the centralized protocol is strictly lower
bounded by that for the distributed protocol and converges to this lower bound
only in the limit of large $K.$
\end{theorem}

\subsection{Outer Bounds}

From the symmetry of the model, it suffices to bound the rate $R_{1}$ of agent
$1$ as
\begin{align}
R_{1}  &  \geq\frac{1}{n}H(J_{1})\geq\frac{1}{n}I(Y_{1}^{n};J_{1}|Y_{2}%
^{n}Y_{3}^{n}\ldots Y_{K}^{n})\\
&  \geq h\left(  Y_{1}|Y_{2}\ldots Y_{K}\right)  -\frac{1}{n}%
%TCIMACRO{\tsum \limits_{i=1}^{n}}%
%BeginExpansion
{\textstyle\sum\limits_{i=1}^{n}}
%EndExpansion
h(Y_{1,i}|\hat{X}_{2,i}Y_{2,i}\ldots Y_{K,i})\label{ROB_3}\\
&  \geq h\left(  Y_{1}|Y_{2}\ldots Y_{K}\right)  -\frac{1}{2}\log(2\pi
e\Sigma) \label{ROB_5}%
\end{align}
where (\ref{ROB_3}) results from the fact that $\hat{X}_{2}^{n},\ldots\hat
{X}_{K}^{n}$ can be estimated from $J_{1},Y_{2}^{n},\ldots Y_{K}^{n}$, and
that conditioning on only one of the estimates is a lower bound on $R_{1},$
and (\ref{ROB_5}) results from using the fact that a jointly Gaussian
distribution maximizes the differential entropy for a fixed variance, from the
concavity of the $\log$ function for $\Sigma\equiv E\left[  var\left(
Y_{1}|\hat{X}_{1}Y_{2}Y_{3}\ldots Y_{K}\right)  \right]  .$ For jointly
Gaussian $\left(  Y_{1},\ldots,Y_{K},\hat{X}_{2}\right)  ,$ we can write
\begin{equation}
\hat{X}_{2}=Y_{2}+%
%TCIMACRO{\tsum \nolimits_{l=1,l\not =2}^{K}}%
%BeginExpansion
{\textstyle\sum\nolimits_{l=1,l\not =2}^{K}}
%EndExpansion
bY_{l}+Z \label{Xhat_OB}%
\end{equation}
where $Z\sim N\left(  0,\sigma_{Z}^{2}\right)  $ is independent of $Y_{k}$ for
all $k,$ and from symmetry, we choose the same scaling constant $b$ in
(\ref{Xhat_OB}). For $g\equiv E[\left(  \hat{X}_{2}-Y_{2}-bY_{3}\ldots
-bY_{K}\right)  ^{2}]=b^{2}/\left(  b^{2}\alpha+\sigma_{Z}^{2}\right)  $,
$c_{1}=\beta^{2}g,$ and $c_{2}=c_{1}+\left(  \beta-\beta\alpha g\right)
^{2}/\left(  \alpha-\alpha^{2}g\right)  ,$ we obtain
\begin{align}
R_{1}  &  \geq\frac{1}{2}\log\left(  \frac{f_{1}(K,\beta^{2}/\alpha)\left(
\alpha-\beta\right)  }{f_{1}\left(  K-1,\beta^{2}/\alpha\right)  }\right) \\
&  -\frac{1}{2}\log\left(  \frac{f_{1}\left(  K,c_{2}\right)  }{f_{1}\left(
K,c_{1}\right)  }\left(  \alpha-\alpha^{2}g\right)  \right)
\end{align}
where we have used the orthogonality of the minimum MSE estimate and the
measurements, i.e., $E\left[  \left(  X_{1}-\hat{X}_{1}\right)  Y_{l}\right]
=0,$ for all $l\not =1,$ and the distortion constraint in (\ref{Dist}).

With $\hat{X}_{2}$ in (\ref{Xhat_OB}), one can similarly bound $L_{1}^{\left(
j\right)  }=L_{1}^{\left(  2\right)  }$ (from symmetry), for all $j$, as
\begin{align}
R_{1}  &  \geq\frac{1}{n}I(X_{1}^{n};Y_{2}^{n}J_{1}J_{2}\ldots J_{K})\\
&  \geq h\left(  X_{1}\right)  -\frac{1}{2}\log\left(  2\pi eE\left[
var\left(  X_{1}|Y_{2}\hat{X}_{2}\right)  \right]  \right) \\
&  =\frac{1}{2}\log\left(  q_{1}\left/  \left(  \left(  1-\sigma_{X}^{2}%
q_{2}^{2}\right)  q_{1}-\sigma_{X}^{2}\left(  \sqrt{h}-q_{2}\right)
^{2}\right)  \right.  \right) \nonumber
\end{align}
where $g_{1}\equiv$ $E\left[  \left(  \hat{X}_{2}-Y_{2}\right)  ^{2}\right]  $
$\ =$ $(b^{2}\left(  K-1\right)  \alpha$ $+$ $\left(  K-1\right)  \left(
K-2\right)  b\beta/2+\sigma_{Z}^{2-1})^{-1},$
\begin{align}
q_{1}  &  \equiv\alpha-g_{1}b^{2}\beta^{2}\left(  K-1\right)  ^{2},\text{
and}\\
q_{2}  &  =g_{1}b^{2}\left(  1+\left(  K-2\right)  \sqrt{h}\right)
\beta\left(  K-1\right)  .
\end{align}

\begin{remark}
Due to the lack of a pre-log factor $K,$ the per-user rate $R$ for the outer
bound rapidly approaches $0$ with $K$ (relative to the inner bounds).
\end{remark}

The rate $R$ and leakage $L_{k}$ (for any $k)$ as a function of $K$ are
illustrated in Fig. \ref{Fig1} for $h=0.5$ and $\sigma_{Q}^{2}=6$.%

%TCIMACRO{\FRAME{ftbpFU}{2.9879in}{2.6333in}{0pt}{\Qcb{Plot of per-user rate
%$R$ and leakage $L_{k}$ of any agent $k$ vs. $K$.}}{\Qlb{Fig1}}%
%{achievaberateleak.eps}{\special{ language "Scientific Word";
%type "GRAPHIC";  maintain-aspect-ratio TRUE;  display "USEDEF";
%valid_file "F";  width 2.9879in;  height 2.6333in;  depth 0pt;
%original-width 8.0998in;  original-height 6.8139in;  cropleft "0.0291";
%croptop "0.9721";  cropright "0.9415";  cropbottom "0.0171";
%filename 'AchievabeRateLeak.eps';file-properties "XNPEU";}} }%
%BeginExpansion
\begin{figure}
[ptb]
\begin{center}
\includegraphics[
trim=0.235704in 0.116518in 0.473838in 0.190108in,
height=2.6333in,
width=2.9879in
]%
{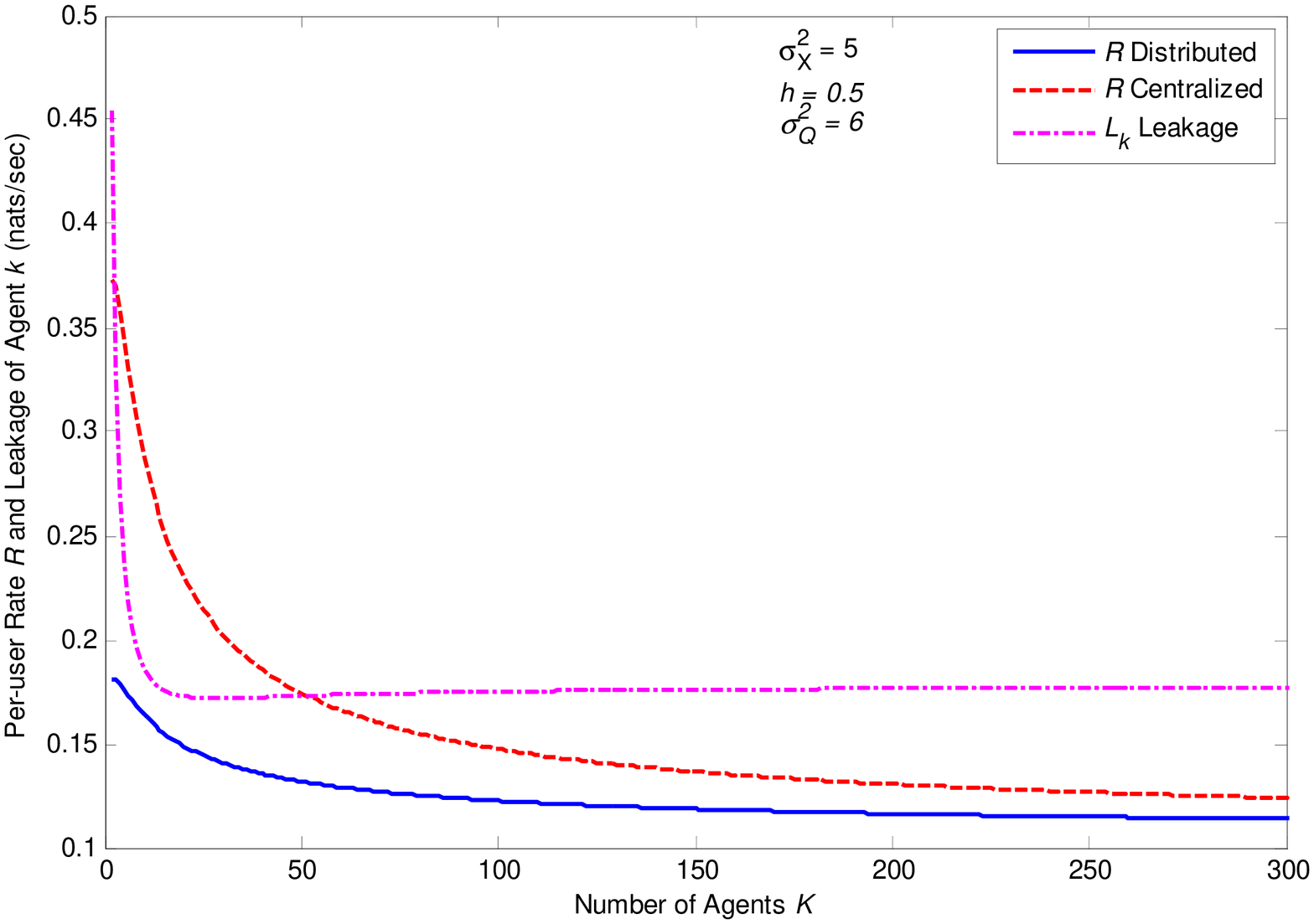}%
\caption{Plot of per-user rate $R$ and leakage $L_{k}$ of any agent $k$ vs.
$K$.}%
\label{Fig1}%
\end{center}
\end{figure}
%EndExpansion

\section{\label{Sec_CR}Concluding Remarks}

We have introduced a distributed state estimation problem among $K$ agents
with fidelity and privacy constraints. We have shown that the sum-rate and per
user rate achieved from a distributed protocol in which the agents directly
interact taking into account the prior knowledge at all agents lower bounds
those achieved by a centralized protocol with convergence for very large $K.$
Tighter outer bounds that account for the distributed coding are much needed.

\bibliographystyle{IEEEtran}
\bibliography{ascexmpl,refsSG}

\end{document}